\documentclass[11pt,final,letterpaper]{scrartcl}
\usepackage{microtype}
\usepackage{caption}
\usepackage{subcaption}
\usepackage{amsmath,amsthm}
\usepackage{amssymb}
\usepackage{thmtools}
\usepackage{thm-restate}
\usepackage{subcaption}
\usepackage{circuitikz}
\usepackage{tikz}
\usepackage{pgfplots}

\newcommand{\Rset}{\mathbb{R}}

\newcommand{\Nset}{\mathbb{N}}

\theoremstyle{plain}
\newtheorem{theorem}{Theorem}[section]
\newtheorem{cor}[theorem]{Corollary}

\newtheorem{lemma}[theorem]{Lemma}
\theoremstyle{definition}

\newcommand{\ceil}[1]{\left\lceil #1 \right\rceil}
\newcommand{\ld}{\log_{2}}
\newcommand{\loq}{\log_{\varphi}}

\newcommand{\floor}[1]{\left\lfloor #1 \right\rfloor}
\newcommand{\sset}[1]{\left\{#1\right\}}
\newcommand{\nneg}[1]{\left(#1\right)^{+}}
\newcommand{\vare}{\varepsilon}

\definecolor{darkRed}{rgb}{0.6,0,0}
\definecolor{lightRed}{rgb}{1,0.75,0.75}
\definecolor{darkGreen}{rgb}{0,0.5,0}
\definecolor{PineGreen}{rgb}{0.01,0.5,0.45}
\definecolor{darkBlue}{rgb}{0,0,0.75}
\definecolor{CornflowerBlue}{rgb}{0.15,0,0.7}
\definecolor{lightBlue}{rgb}{0.75,0.75,1}
\definecolor{grey}{rgb}{0.5,0.5,0.5}
\definecolor{black}{rgb}{0,0,0}
\definecolor{red}{rgb}{1,0,0}
\definecolor{green}{rgb}{0,1,0}
\definecolor{blue}{rgb}{0,0,1}
\definecolor{yellow}{rgb}{1,1,0}
\definecolor{orange}{rgb}{1,0.6,0}
\definecolor{cyan}{rgb}{0,0.7,1}
\definecolor{purple}{rgb}{0.5,0,0.8}

\usetikzlibrary{shapes}
\usetikzlibrary{arrows,shapes.gates.logic.US,shapes.gates.logic.IEC,calc,backgrounds}
\usetikzlibrary{patterns}

\makeatletter \tikzset{circle split part fill/.style args={#1,#2}{%
    alias=tmp@name,
    postaction={%
      insert path={ \pgfextra{%
          \pgfpointdiff{\pgfpointanchor{\pgf@node@name}{center}}%
          {\pgfpointanchor{\pgf@node@name}{east}}%
          \pgfmathsetmacro\insiderad{\pgf@x} \fill[#1]
          (\pgf@node@name.base)
          ([xshift=-\pgflinewidth]\pgf@node@name.east) arc
          (0:180:\insiderad-\pgflinewidth)--cycle; \fill[#2]
          (\pgf@node@name.base)
          ([xshift=\pgflinewidth]\pgf@node@name.west) arc
          (180:360:\insiderad-\pgflinewidth)--cycle; }}}}}
\makeatother

\title{Binary Adder Circuits of  Asymptotically Minimum Depth, Linear Size, and Fan-Out Two}

\author{Stephan Held \\ Research Institute for Discrete Mathematics, University of Bonn\\ \\Sophie Theresa Spirkl \\ Research Institute for Discrete Mathematics, University of Bonn}
\begin{document}

\markboth{S. Held and S. Spirkl}{Binary Adder Circuits of  Asymptotically Minimum Depth, Linear Size, and Fan-Out Two}

\maketitle
\begin{abstract}
  We consider the problem of constructing fast and small binary adder
  circuits. Among widely-used adders, the Kogge-Stone adder is often
  considered the fastest, because it computes the carry bits for two
  $n$-bit numbers (where $n$ is a power of two) with a depth of $2\ld
  n$ logic gates, size $4 n\ld n$, and all fan-outs bounded by
  two. Fan-outs of more than two are avoided, because they lead to the
  insertion of repeaters for repowering the signal and additional
  depth in the physical implementation.

  However, the depth bound of the Kogge-Stone adder is off by a factor
  of two from the lower bound of $\ld n$. This bound is achieved
  asymptotically in two separate constructions by Brent and
  Krapchenko. Brent's construction gives neither a bound on the
  fan-out nor the size, while Krapchenko's adder has linear size, but
  can have up to linear fan-out. With a fan-out bound of two, neither
  construction achieves a depth of less than $2 \ld n$.

  In a further approach, Brent and Kung proposed an adder with
  linear size and fan-out two, but twice the depth of the Kogge-Stone adder.

  These results are 33-43 years old
  and no substantial theoretical improvement for has been made since then.
  In this paper we integrate the individual advantages of all previous
  adder circuits into a new family of full adders, the first to
  improve on the depth bound of $2\ld n$ while maintaining a fan-out
  bound of two. Our adders achieve an asymptotically optimum logic
  gate depth of $\ld n + o(\ld n)$ and linear size $\mathcal {O}(n)$.
\end{abstract}

\section{Introduction }

\label{introduction}

Given two binary addends $A = (a_n \dots a_1)$ and $B = (b_n \dots b_1)$,
where index $n$ denotes the most significant bit, their sum $S = A+B$
has $n+1$ bits.
We are looking for a logic circuit, also called an \emph{adder}, that
computes $S$.  Here, a \emph{logic circuit} is a non-empty connected
acyclic directed graph consisting of nodes that are either
\emph{gates} with incoming and outgoing edges, \emph{inputs} with at
least one outgoing edge and no incoming edges, or \emph{outputs} with
exactly one incoming edge and no outgoing edges.  Gates represent one
or two bit Boolean functions, specifically \textsc{And}, \textsc{Or},
\textsc{Xor}, \textsc{Not} or their negations.  A small example is
shown on the right side of Figure~\ref{fig:pfx-to-gate3}. The
\emph{fan-in} is the maximum number of incoming edges at a vertex, and
it is bounded by two for all gates.

The main characteristics in adder design are the {\it depth}, the {\it size},
and the {\it fan-out} of a circuit.
The depth is defined as the maximum length of a directed path in the logic
circuit and is used as a measure for its speed. The lower the
depth, the faster  is the adder. The size is the total
number of gates in the circuit, and is used as a measure for the space and power
consumption of the adder, both of which we aim to minimize. The
fan-out is the maximum number of outgoing edges at a vertex. High
fan-outs increase the delay and require additional repeater gates (implementing the identity function) in physical design.
Thus, when comparing the depth of adder circuits, their fan-out should be considered as well; we will focus on the usual fan-out bound of two. Circuits with higher fan-outs can be transformed into fan-out two circuits by replacing each interconnect with high fan-out by a balanced binary {\it repeater tree}, \.i.e. the underlying graph is a tree and all gates are repeater gates. However, this increases the size linearly and the depth logarithmically in the fan-out.
Hoover, Klawe, and Pippenger [1984] gave a smarter way to bound the fan-out of a given circuit,
but it would also triple the size and depth in our case of gates with two inputs.

Using logic circuit depth as a measure for speed
 is a common practice in logic synthesis that simplifies many aspects of physical hardware.
In CMOS technology, \textsc{Nand}/\textsc{Nor} gates are faster than
\textsc{And}/\textsc{Or} gates and efficient implementations exist for
integrated multi-input {\sc And-Or}-Inversion gates and {\sc Or-And}-Inversion
gates.
We assume that a {\it technology mapping} step
\cite{Chatterjee+techmap2006,keutzer88} translates the adder circuit
after logic synthesis using logic gates that are best for the given
technology.
Despite its simplicity, the depth-based model is at the core of
programs such as {\sc BonnLogic} \cite{bonnlogic} for refining carry bit circuits,
which is an integral part of the current IBM microprocessor design flow.
Recently, we reduced the running time for computing such carry bit circuits
significantly from $\mathcal{O}(n^3)$ to $\mathcal{O}(n \log n)$
\cite{held+spirkl:2014}.
Exemplary, for all newly proposed adder circuits in this paper  we will demonstrate
how to efficiently transform them into equivalent  circuits 
using only  \textsc{Nand}/\textsc{Nor} and \textsc{Not} gates. 

Like most existing adders, we use the notion of generate and propagate signals, see  \cite{sklansky,brent,knowles}.
For each position $1 \leq i \leq n$, we compute a
\emph{generate signal} $y_i$ and a \emph{propagate signal}
$x_i$, which are defined as follows:
\begin{equation}
\begin{array}{rl}
x_i &= a_i \oplus b_i,\\
y_i &= a_i \wedge b_i,
\end{array}
\label{eqn:generate-and-propagat}
\end{equation}
where $\wedge$ and $\oplus$ denote the binary \textsc{And} and
\textsc{Xor} functions.  The \emph{carry bit} at position $i+1$ can be
computed recursively as $c_{i+1} = y_i \vee (x_i \wedge c_i)$, since
there is carry bit at position $i+1$ if the $i$-th bit of both inputs
is $1$ or, assuming this is not the case, if at least one (hence
exactly one) of these bits is $1$ and there was a carry bit at
position $i$.

The first carry bit $c_1$ can be used to represent the
\emph{carry-in}, but we usually assume $c_1 = 0$. The last carry bit
$c_{n+1}$ is also called the \emph{carry-out}. From the carry bits, we
can compute the output $S$ via
\begin{equation}
s_i = c_i \oplus x_i \text{ for }  1 \leq i \leq n \text{ and }  s_{n+1} = c_{n+1}.
\label{eqn:sum-computation}
\end{equation}
With this preparation of constant depth, linear size, and fan-out two at the inputs $a_i, b_i$ and fan-out one at the carry bits $c_{i+1}$ ($i=1,\dots,n$), the binary addition is reduced to the problem of
computing all carry bits $c_{i+1}$  from $x_i,y_i$
($i=1,\dots,n$).

{\flushleft{\bf Convention}: }
{\it From now on, we will omit the preparatory steps (\ref{eqn:generate-and-propagat}) and (\ref{eqn:sum-computation}) and
consider a circuit an adder circuit if it  computes all $c_{i+1}$  from $x_i,y_i$ ($i=1,\dots,n$).}
\vspace*{0.5cm}

Expanding the recursive formula for $c_{i+1}$ as in equation  (\ref{eqn:and-or-form}) results in  a logic circuit that is a path of alternating \textsc{And} and \textsc{Or} gates. It corresponds to the long addition method and has linear depth $2(n-1)$.

\begin{eqnarray}
c_{i+1} = \ &y_i \vee \left(x_i \wedge (y_{i-1} \vee (x_{i-1} \wedge \dots \wedge (y_2 \vee (x_2 \wedge y_1)).
\dots ))\right)  \label{eqn:and-or-form} 
\end{eqnarray}

\subsection{Prefix Graph Adders}
\label{sec:prefix-gate-adders}
For two pairs $z_i = (x_i,
y_i)$ and $z_j = (x_j, y_j)$, we define the associative {\it prefix operator} $\circ$
as
\begin{equation}
{x_i \choose y_i} \circ {x_j \choose y_j} = {x_i \wedge x_j \choose y_i \vee (x_i \wedge y_j)}.
\label{eqn:prefix-operator}
\end{equation}
A circuit computing (\ref{eqn:prefix-operator}) can be implemented as a  logic circuit consisting of three gates and with depth two as shown in Figure~\ref{fig:pfx-to-gate3}.
For $i=1,\dots,n$, the result of the prefix computation $z_i \circ \dots \circ z_1$ of the expression $z_n \circ \dots \circ z_1$ contains the carry bit $c_{i+1}$:
\begin{equation}
{x_i \wedge x_{i-1} \wedge \dots \wedge x_1 \choose c_{i+1}} = {x_i
  \choose y_i} \circ {x_{i-1} \choose y_{i-1}} \circ \dots \circ {x_1
  \choose y_1}.\label{eqn:prefix-gate-carry}
\end{equation}

A circuit of $\circ$-gates computing all prefixes $z_i \circ \dots \circ z_1$ ($i=1,\dots,n$) for an associative operater $\circ$ is called a \emph{prefix graph}. A prefix graph yields an adder by expanding each $\circ$-gate as in Figure~\ref{fig:pfx-to-gate3}, and extracting the carry bits as in~\eqref{eqn:prefix-gate-carry}. 

\begin{figure}
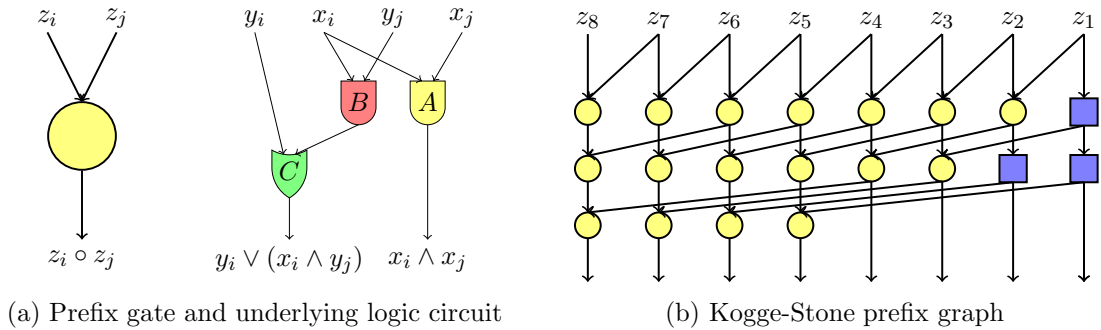

  \begin{subfigure}[b]{0.49\linewidth}
    
  \centering{
    \resizebox{0.85\linewidth}{!}{
      \begin{tikzpicture}
        \input{pfx-to-gate3}
      \end{tikzpicture}
    }
    \caption{Prefix gate and underlying logic circuit}
    \label{fig:pfx-to-gate3}%
  }

  \end{subfigure}
  \hfill
  \begin{subfigure}[b]{0.49\linewidth}
    
  \centering{
    \resizebox{1\linewidth}{!}{
      \begin{tikzpicture}
        \input{kogge-pfx}
      \end{tikzpicture}
    }
    \caption{Kogge-Stone prefix graph}
    \label{fig:kogge-pfx}%
  }
    
  \end{subfigure}
\caption{Prefix gate and graph}
\end{figure}

Most previous constructions for adders are based on prefix graphs of small depth, size and/or fan-out.  
Sklansky [1960] developed a prefix graph of minimum depth
$\ld n$, size $\frac{1}{2}n \ld  n$, but high fan-out $\frac{1}{2}n+1$.
The first prefix graph   with logarithmic depth ($2\log n -1$) and linear size  ($3n - \log n - 2$) was developed by Ofman [1962],  exhibiting a non-constant  fan-out of  $\frac{1}{2}\log n$.
Kogge and Stone [1973] introduced the {\it recursive doubling algorithm}
which leads to a prefix graph with depth $\log_2 n$ and fan-out two (see Figure~\ref{fig:kogge-pfx}).
Since we will use variants of it in our construction, we describe it in detail.
For $1\le s\le t\le n$,  let $Z_{s,t} := z_t\circ\dots\circ z_s,$
and for $x \in \Rset$, let $(x)^+:= \max\{x,0\}$.
The graph has ${\ld n}$ levels, and on level $i$ it computes for every input $j$ ($1\le j \le n$)
the prefix $Z_{1+\nneg{j - 2^i},j} = z_j \circ\dots\circ z_{1+\nneg{j - 2^i}}$
according to the recursive formula
\begin{equation}Z_{1+(j - 2^i)^{+},j} = Z_{1+(j - 2^{i-1})^{+} ,j} \circ Z_{1+(j - 2^i)^{+} ,(j - 2^{i-1})^{+}}, \label{eq:ks}
\end{equation}
from the prefixes of sequences of $2^{i-1}$ consecutive inputs computed in the previous level.
The fan-out is bounded by two, since every intermediate result is used exactly twice: once as the ``upper half'' and once as
the ``lower half'' of an expression of the form $z_j \circ \dots \circ
z_{1+\nneg{j - 2^i}}$.
Note that for level $i$ ($1 \leq i \leq \ld n$), a repeater gate (which computes the identity function) is used instead of a $\circ$-gate if $j \leq 2^i$, i.\ e.\ in the case that the right input in \eqref{eq:ks} is empty. Repeaters are shown as blue squares in Figure~\ref{fig:kogge-pfx}.
The Kogge-Stone prefix
graph minimizes both depth and fan-out. On the other hand, since there
is a linear number of gates at each level, the total size in terms of prefix gates is $n{\ld n}
- \frac{n}{2}$.

Ladner and Fischer [1980] constructed a prefix graph of depth $\log_2 n$ but
high fan-out. Brent and Kung found a linear-size prefix graph with fan-out two, but twice the depth of the other constructions.
Finally, Han and Carlson [1987] described a hybrid between a Kogge-Stone adder and a Brent-Kung adder
which achieves a trade-off between depth and size.
Lower bounds for the trade-off between the depth and size of a prefix graph can be found in \cite{fich,sergeev}.

The above prefix graphs can be used for prefix computations with respect to any associative
operator $\circ$.  In fact, we will later use a prefix graph
in which the operator $\circ$ represents an {\sc And} gate.
When turning one of the above prefix
graph adders into a logic circuit for addition such that each prefix gate is implemented as in Figure~\ref{fig:pfx-to-gate3}, the depth of the logic circuit is
twice the depth of the prefix graph and the number of logic gates is
three times the number of prefix gates.
The fan-out of the underlying logic circuit can increase
by one compared to the prefix graph, because the left propagate signal $x_i$ is used twice within a prefix gate.
In Section~\ref{sec:brent-kung-step}, we will see that in the case of the Brent-Kung adder
a fan-out of two can be achieved by using reduced prefix gates.

Any adder constructed from a prefix graph has a logic gate depth of at
least $\loq n -1 > 1.44 \ld n - 1$, where
$\varphi = \frac{1+ \sqrt{5}}{2}$ is the golden section
\cite{held+spirkl:2014}, see also \cite{code-trees}. In
\cite{held+spirkl:2014} an adder of size $\mathcal{O}(n\ld\ld n)$
asymptotically attaining this depth bound is described, however with a
high fan-out of $\sqrt{n}+1$.

\subsection{Non-Prefix Graph Adders}
Since none of the $2n$ inputs $x_i,y_i$ ($1 \le i\le n$) except for $x_1$ are redundant for $c_{n+1}$, the depth of any adder circuit using 2-input gates
 is at least $\ld n+1$, which would be attained by a balanced binary
tree with inputs/leaves  $x_i,y_i$ ($1 \le i\le n$).
With adders that are not based on prefix graphs, this bound is asymptotically tight.
Krapchenko showed that any formula (a circuit with tree topology) for computing $c_{n+1}$   has depth at least $\ld n +0.15\ld\ld\ld n + \mathcal{O}(1)$  \cite{krapchenkoLB}.

Brent [1970] gives an approximation scheme for a single carry bit circuit attaining an
asymptotic depth of $(1+\vare)\ld n +o(\ld n)$ for any given
$\vare>0$.
%
The best known depth for a single carry bit circuit is $\ld n +\ld \ld n + \mathcal{O}(1)$, due to Grinchuk [2008].
However,  \cite{Grinchuk-ShallowCarryBit2009} and \cite{brent} did not address how to overlay circuits for the different
carry bits to bound the size and fan-out of an adder based on their circuits.
One problem in sharing intermediate results is that this can create high fan-outs.

Krapchenko [1967] (see
\cite[pp. 42-46]{wegener}) presented an adder with
asymptotically optimum depth $\ld n +o(\ld n)$ and linear
size.
It was refined for small $n$ by \cite{Gashkov+ImprovingKrapchenko2007}.
However, the fan-out is almost linear.

\subsection{Our Contribution}
\label{sec:our_contribution}

In this paper, we present the first family of adders of asymptotically optimum
depth, linear size, and fan-out bound two:

\begin{restatable}[Main Theorem]{theorem}{maintheorem}
  \label{thm:central-theorem}
  Given two $n$-bit numbers $A$,$B$, there is a logic circuit computing
  the sum $A+B$,  using gates with
  fan-in and fan-out two and that has depth $\ld n + o(\log n)$
  and size $\mathcal{O}(n)$.
\end{restatable}

The rest of the paper is organized as follows.  In
Section~\ref{sec:min-depth-fan-out-two}, we develop a family of adders
of asymptotically minimum depth, fan-out two, but super-linear
size $\mathcal{O}\left(n\ceil{\sqrt{\ld n}}^2 2^{\sqrt{\ld n}}\right)$. In Section~\ref{sec:linearizing-size}, using reductions
similar to \cite{krapchenko}, this adder is transformed into an adder
of linear size with the asymptotically same depth, proving Theorem~\ref{thm:central-theorem}.
While all of the above adders use only \textsc{And}/\textsc{Or} gates and repeaters,
we show in Section~\ref{sec:technology-mapping}
that  Theorem~\ref{thm:central-theorem}
holds also if only   \textsc{Nand}/\textsc{Nor} and \textsc{Not} gates are available.

\section{Asymptotically Optimum Depth  and Fan-Out Two}
\label{sec:min-depth-fan-out-two}
For $1\le s\le t \le n$, let     $X_{s,t}$ and $Y_{s,t}$ denote the propagate and generate signal
for the sequence of indices between $s$ and $t$, i.e.\
\begin{equation}
  \begin{array}{rl}
    X_{s,t} &= \bigwedge_{i=s}^t x_i\\
    Y_{s,t} &= y_t \vee \left(x_t \wedge (y_{t-1} \vee (x_{t-1} \wedge \dots \wedge (y_{s+1} \vee (x_{s+1} \wedge y_s))\dots ))\right)\\
  \end{array}
\end{equation}

The adders based on prefix graphs as in Section~\ref{sec:prefix-gate-adders} impose
a common topological structure on the computation of intermediate results $X_{s,t}$  and
$Y_{s,t}$.
In the adder described by Brent [1970], on the other hand, intermediate results $X_{s,t}$ and $Y_{s,t}$ are computed separately within larger blocks.


Let $n=2^{rk}$ for $r \in \mathbb{N}$ and $k \in \mathbb{N}$ to be chosen later.
A central idea of generating a faster adder is to use multi-fan-in (also called high-radix) subcircuits within a  Kogge-Stone prefix graph.
While all the prefix gates in Figure~\ref{fig:kogge-pfx} have fan-in two, we want to use prefix gates with fan-in  $2^r$,
so that the number of levels reduces  from $\ld n$ to $\log_{2^r} n$ = $\frac{1}{r}\ld n$.
Each prefix gate with fan-in $2^r$ represents a logic circuit with fan-in and fan-out bounded by two. Since the output of each prefix gate will be used in $2^r$ prefix gates at the next level, our approach also requires to duplicate the intermediate result at the output of a prefix gate $2^{r-1}$ times.
To accomplish this, we consider the computation of generate and
propagate sequences separately.

Our adder consists of two global Kogge-Stone type prefix graphs. The
first such graph uses 2-input \textsc{And}-gates and computes
propagate signals used in the other prefix graph. This graph uses
$2^r$-input subcircuits that are arranged in the same way as the
Kogge-Stone graph, and it computes the generate (carry) signals. 
Both
graphs are modified to duplicate intermediate generate signals  $2^{r-1}$ times 
and intermediate propagate signals  $2^r$
times so that the overall constructions obeys the fan-out bound of two.

\subsection{Multi-Input  Generate Gates}
\label{sec:multi-input-generate-gate}
We now introduce {\it multi-input generate gates}, which are the main
building block for computing the generate signals.
%
%
Given  $2^r$ propagate and generate pairs $(\tilde{x}_{2^r},\tilde{y}_{2^r}),\dots, (\tilde{x}_{1},\tilde{y}_{1})$,
a multi-input generate gate computes the generate signal
$$\tilde{Y}_{1,2^r} = \tilde{y}_{2^r} \vee \left(\tilde{x}_{2^r} \wedge (\tilde{y}_{2^r-1} \vee (\tilde{x}_{2^r-1} \wedge \dots \wedge (\tilde{y}_{2} \vee (\tilde{x}_{2} \wedge \tilde{y}_1))\dots ))\right).$$
The input pairs $(\tilde{x}_i,\tilde{y}_i)$ $(i\in \{1,\dots, 2^r\})$ are not necessarily the input pairs  of the adder; they can be intermediate results.

Each multi-input generate gate has $2^{r-1}$ outputs, each of which provides the result  $\tilde{Y}_{1,2^r}$, because later we want to reuse this signal $2^r$ times and bound the fan-out of each output by two. In contrast to two-input prefix gates computing (\ref{eqn:prefix-operator}), multi-input generate
gates do not compute the propagate signals $\tilde{X}_{1,2^r}=\bigwedge_{i=1}^{2^r} \tilde{x}_i$ for the given input pairs.
 All required propagate signals will be computed by
the separate {\sc And}-prefix graph, described in
Section~\ref{sec:and-pfx-graph}.

\begin{figure}
 
  \centering{
    \resizebox{0.5\linewidth}{!}{
      \begin{tikzpicture}
        \input{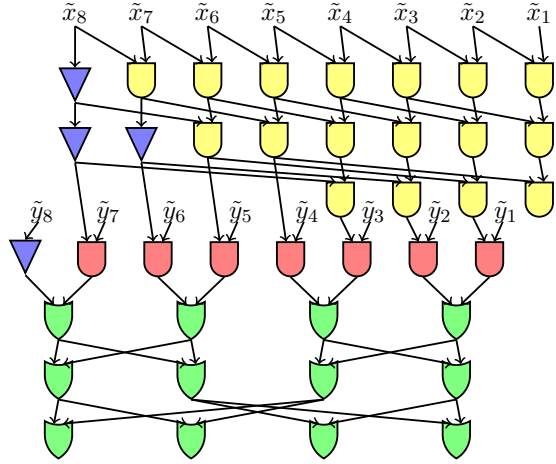}
      \end{tikzpicture}
    }
    \caption{A $2^r$-input $2^{r-1}$-output generate gate for $r=3$}
    \label{fig:fo-pfx-small}%
  }

\end{figure}
Figure~\ref{fig:fo-pfx-small} shows an example of a multi-input generate gate with $8$ inputs. A $2^r$-input prefix gate computes $\tilde{Y}_{1,2^r}$ as in the disjunctive normal form
$$ \tilde{Y}_{1,2^r} = \displaystyle  \bigvee_{j=1}^{2^r} \left( \tilde{y}_j \wedge\left(\bigwedge_{i=j+1}^{2^r} \tilde{x}_i\right)\right),$$
first computing all the minterms $m_j:=\tilde{y}_j \wedge\left(\bigwedge_{i=j+1}^{2^r} \tilde{x}_i\right)$ $(j=1,\dots,2^r)$, and then
the disjunction $\bigvee_{j=1}^{2^r} m_j$.
The terms $\bigwedge_{i=j+1}^{2^r} \tilde{x}_i$ are computed as a Kogge-Stone {\sc And}-suffix graph, which arises
from a Kogge-Stone prefix graph by reversing the ordering of the inputs.
A single stage of (red) {\sc And}-gates and one repeater concludes the computation of the minterms. Each input $\tilde{y}_i$ is used exactly once within this circuit. The repeater is dispensable but simplifies the size formula and will become useful in Section~\ref{sec:technology-mapping}.

Finally, instead of computing the disjunction $\bigvee_{j=1}^{2^r}
m_j$ by a balanced binary {\sc Or} tree and duplicating the results
$2^{r-1}$ times through a balanced repeater tree, the duplication is accomplished by $r$
rows of $2^{r-1}$ {\sc Or}-gates as shown in Figure~\ref{fig:fo-pfx-small}. Formally, let $M_{i,j} = \bigvee_{i'=i}^{j} m_{i'}$ be the conjunction of minterms $i, i+1, \dots, j$. Then, on level $l \in \{1,\dots, r\}$, we compute each signal of the form $M_{i2^{l} + 1, (i+1)2^l}$, $i=0, \dots, 2^{r-l}-1$, from the previous level, and we compute $2^{l-1}$ copies of it. By using $M_{i2^{l} + 1, (i+1)2^l} = M_{{2i 2^{l-1} + 1, (2i+1)2^{l-1}}} \vee M_{{(2i+1)2^{l-1} + 1, (2i+2)2^{l-1}}}$, and since each preceding signal is available $2^{l-2}$ times ($l\ge 2$), we can ensure that each of them has fan-out two. On the last level, we will have computed $2^{r-1}$ copies of $M_{1,2^r} = \tilde{Y}_{1,2^r}$. Each level uses $2^{r-1}$ \textsc{Or}-gates.
Note that a similar construction for reducing fan-out has been used by Lupanov
when  extending his well-known  bounded-size representation of general Boolean functions to circuits  with bounded fan-out \cite{LupanovBoundedBranching62}.

\begin{lemma}
  The multi-input generate gate has   $2^r$ generate/propagate pairs as input and
  $2^{r-1}$ outputs. Each propagate input has fan-out two and each generate input has fan-out one.
  The gate consists of $r2^r + (r+1)2^{r-1}$ logic gates
  which have  fan-out at most two.
  The depth for the propagate inputs $\tilde{x}_i$  is $2r+1$
  and the depth for the generate inputs $\tilde{y}_i$ is $r+1$
  ($i\in \{1,\dots,2^r\}$).
\end{lemma}
\begin{proof}
  All the terms $\bigwedge_{i=j+1}^{2^r} \tilde{x}_i$ are computed as
  a Kogge-Stone {\sc And}-suffix graph (blue and yellow gates in
  Figure~\ref{fig:fo-pfx-small}) of size
  $$2^r\lceil \ld 2^r\rceil -\frac{2^r}{2} = (r-1)2^r + 2^{r-1}.$$

  Then, there is a level of $2^r$ (red) {\sc And} gates and one
  repeater, concluding the computation of the minterms.  Finally,
  there are $r2^{r-1}$ (green) {\sc Or}-gates to compute the
  disjunction $\bigvee_{j=1}^{2^r} m_j$ $2^r$ times, for a total of
  $$r2^r + (r+1)2^{r-1} $$ gates.
  By construction, no gate and propagate input has fan-out larger than two, and all generate  inputs have fan-out one.
  The depth is $r$ for the  {\sc And}-suffix graph,
  one  for the red gates, and $r$ for the  disjunctions, yielding the desired depths of $2r+1$ for the propagate inputs
  and $r+1$ for the generate inputs.
\end{proof}

\subsection{Augmented Kogge-Stone  \textsc{And}-Prefix Graph}
\label{sec:and-pfx-graph}
The second important component of our construction is the {\it augmented
  Kogge-Stone \textsc{And}-prefix graph}.  It is used to compute
$X_{s,t}= \bigwedge_{i=s}^t x_i$ for all $1 \leq t \leq n$ and $s
=1+\nneg{t-2^{rl}}$ with $0 \leq l < k$, providing each output
$2^r$ times through $2^r$ individual gates.
It is constructed as follows.  First, we take a Kogge-Stone
[1973] prefix graph, where the prefix operator is an {\sc And}-gate, i.e.\ $\circ = \wedge$.  It consists of
$\ld n$ levels, and on level $i$ it computes for every input $j$
($1\le j \le n$) the prefix $X_{1+\nneg{j - 2^i},j}$
from the prefixes of sequences of $2^{i-1}$ consecutive inputs computed in the previous level.

Each of the results $X_{s,t}$ from level $rl$ will later be used in  $2^r$ multi-input generate gates
for all $0\le l < k$, $s =1+ \nneg{t-2^{rl}}$ and $ 1\le t \le n$.
In order to achieve a fan-out bound of two, starting at the inputs, we insert one row of $n$ repeaters after every $r$ levels of {\sc And}-gates. This allows to use the repeaters as the inputs for the next level, and to extract the signals $X_{s,t}$ once at the \textsc{And}-gates before the repeaters. The construction is shown in Figure~\ref{fig:kogge-aug} with the extracted outputs $X_{s,t}$ shown as red arrows.
\begin{figure}

  \centering{
    \resizebox{1.0\linewidth}{!}{
      \begin{tikzpicture}
        \input{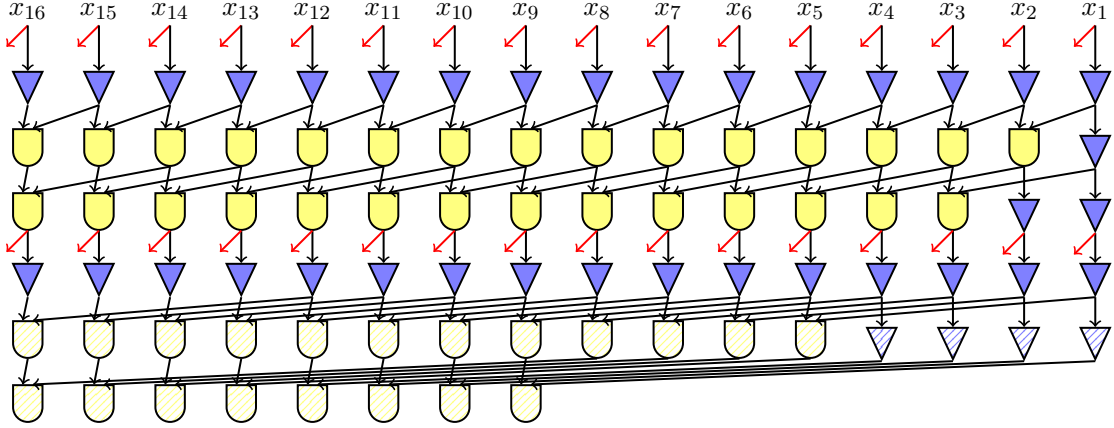}
      \end{tikzpicture}
    }
    \caption{Augmented Kogge-Stone \textsc{And}-prefix graph for $r=k=2$.}
    \label{fig:kogge-aug}%
  }

\end{figure}
The last block of $r$ rows of gates  (hatched gates in Figure~\ref{fig:kogge-aug})  of the Kogge-Stone prefix graph can be omitted in our construction to reduce the size.

Each  output signal $X_{s,t}$ will be input to a multi-input generate gate,
where it is immediately duplicated. Thus, each output $X_{s,t}$ 
of the augmented Kogge-Stone \textsc{And}-prefix graph has to be provided through an individual gate.
To this end, at each of the $nk$ outputs, we add $2^{r+1}-1$ repeater gates as the vertices of a
balanced binary tree to create $2^r$ copies of the signal with a single repeater serving each leaf. For simplicity
these repeaters are hidden in Figure~\ref{fig:kogge-aug}.

\begin{lemma}\label{lem:size-augmented-kogge-stone-and-prefix-graph}
The total size of the augmented Kogge-Stone \textsc{And}-prefix graph is $ nr(k-1) + nk2^{r+1}$.
\end{lemma}
\begin{proof}
Each  binary repeater tree at one of  the $nk$ outputs consists of  $2^{r+1}-1$ repeaters, summing up to $nk(2^{r+1}-1)$
repeaters in these repeater trees.
The remaining circuit consists of  $r(k-1)+k$ rows ($r(k-1)$ rows of {\sc And}-gates and $k$ rows of repeaters) of $n$ gates each,
summing up to $n(r(k-1)+k)$ gates.
Altogether, the circuit contains $ nr(k-1) + nk2^{r+1}$ gates.
\end{proof}

\begin{lemma}
The signal $X_{s,t}$ for   $1 \leq t \leq n$  and $s = 1+ \nneg{t-2^{rl}}$ for $0 \leq l < k$
is available $2^r$  times with internal fan-out one  at a  depth of $(l+1)(r+1)$.
\end{lemma}
\begin{proof}
  The functional correctness is clear by construction. For the depth
  bound, let $1 \leq t \leq n$ and $0 \leq l < k$. Then, for
  $s =1+\nneg{t-2^{rl}}$, the signal $X_{s,t}$ is available at the
  bottom of the $l$-th block at a depth of $l(r+1)$.  Subsequently, we
  create $2^r$ copies of the signal in a repeater tree of depth $r+1$.
  Together, this gives the desired depth $(l+1)(r+1)$.
\end{proof}

\subsection{Multi-Input Generate Adder}
We now describe the multi-input generate adder for $n = 2^{rk}$.
It consists of an augmented Kogge-Stone {\sc And}-prefix graph from the previous
section and a circuit consisting of multi-input generate gates similar
to a radix-$2^r$ Kogge-Stone adder.

The construction uses $k$ rows with $n$ multi-input generate gates or repeater trees  (see Figure~\ref{fig:fixed-size2}).
The $t$-th multi-input generate gate in level $l\in\{1,\dots,k\}$ computes $Y_{1+\nneg{t-2^{rl}},t}$ according to the formula $Y_{1+\nneg{t-2^{rl}},t}  = $
\begin{equation}
\bigvee_{j=1}^{2^r} \left( Y_{1+\nneg{t-j2^{r(l-1)}},\nneg{t-(j-1)2^{r(l-1)}}} \wedge \left(\bigwedge_{k=j+1}^{2^r}X_{1+\nneg{t-k2^{r(l-1)}},\nneg{t-(k-1)2^{r(l-1)}}} \right)\right).
\label{eqn:multi-input-generate-adder-recursion}
\end{equation}

If $\nneg{t-2^{rl}} < \nneg{t-2^{r(l-1)}}$ (yellow circuits in Figure~\ref{fig:fixed-size2}),
this computation is carried out using a multi-input generate gate from Section~\ref{sec:multi-input-generate-gate}. 
As its inputs, it uses generate signals from the previous level,
$l-1$, and  propagate signals obtained from the augmented Kogge-Stone
\textsc{And}-prefix graph.

Except for the last level, each intermediate generate signal will be used $2^r$ times as in (\ref{eqn:multi-input-generate-adder-recursion})
in the next level. As the fan-out of each generate input inside a multi-input generate gate is one,
we need to provide $2^{r-1}$ copies through individual gates to serve  $2^r$  multi-input generate gates with fan-out two.

If $\nneg{t-2^{rl}} = \nneg{t-2^{r(l-1)}}$ (blue squares in
Figure~\ref{fig:fixed-size2}), $Y_{1+\nneg{t-2^{rl}},t}$ is already
computed in the previous level, and in this level it is sufficient to
duplicate the signal $2^{r-1}$ times using a balanced binary repeater tree.

The augmented Kogge-Stone {\sc And}-prefix graph provides each signal $2^r$ times with individual repeaters.
Thus, it can be distributed to $2^r$ multi-input generate gates, where the fan-out of each propagate input is two.

For the first level of multi-input generate gates, we duplicate each generate signal $y_i$ at an input $i\in
\{1,\dots,n\}$ using a balanced binary repeater tree of depth $r-1$
and size $2 + 2^2 + \dots + 2^{r-1} = 2^{r}-2$. Again, we can distribute each copy to two multi-input generate gates,
maintaining fan-out two.

In the last level of multi-input generate gates, we do not need to duplicate the signals
anymore. Instead of the $r$ rows of $2^{r-1}$ {\sc Or}-gates each, we can compute the single outputs
using a balanced binary tree of $2^r-1$ {\sc Or}-gates  and depth $r$.

\begin{figure}
  
  \centering{
    \resizebox{1\linewidth}{!}{
      \begin{tikzpicture}
        \input{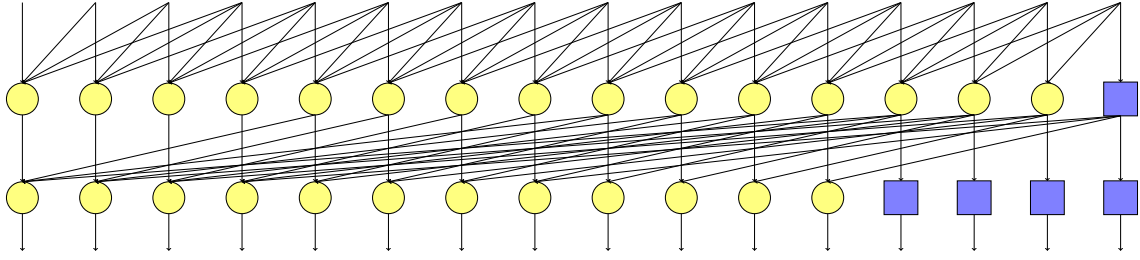}
      \end{tikzpicture}
    }
    \caption{Multi-input multi-output generate  gate adder for $r=k=2$}
    \label{fig:fixed-size2}%
  }

\end{figure}
\begin{lemma}
The  multi-input generate adder for $n=2^{rk}$ bits  obeys a fan-out bound of two,  contains less than
$$3  nk(r+2)2^{r-1} + n2^{r} +    nrk$$
gates,  and has depth
    $$ kr + 2r + k + 1.$$
\label{lem:depth+size-lemma}
\end{lemma}
\begin{proof}
  Inside  each multi-input generate gate, the fan-out of propagate inputs is two and the fan-out of generate inputs is one.
  Thus, it suffices to observe that in each non-output level there are $2^{r}$ copies of each propagate signal and 
  $2^{r-1}$ copies of each generate signal, and that the 
  fan-out of two holds within the augmented Kogge-Stone graph and within each multi-input generate gate.

  By Lemma~\ref{lem:size-augmented-kogge-stone-and-prefix-graph}, the
  size of the augmented Kogge-Stone {\sc And}-prefix graph is $nr(k-1) + nk2^{r+1}$.
  The size of the $n$ balanced binary trees duplicating the input generate signals is $n (2^{r}-2)$.

  The remainder of the graph consists of $k$ rows of $n$ $2^r$-input multi-input generate gates or repeater trees.
  The size of a repeater tree (blue boxes in Figure~\ref{fig:fixed-size2})   is at most  $2^{r-1}-1 \le r2^r + (r+1)2^{r-1}$ ($r\ge 1$), which is
  the size of a multi-input generate gate.
  Thus, the size of all these multi-input generate gates is at most $nk\left(r2^r + (r+1)2^{r-1}\right)$.
  Summing up, the total size is at most
  \begin{equation*}    
    \begin{array}{rl}
      & nr(k-1) + nk 2^{r+1} + n (2^{r}-2) + nk \left(r2^r + (r+1)2^{r-1}\right) \\
      =  & nk2^{r+1} + nkr2^r + n2^r + nk(r+1)2^{r-1} + nkr - n(r + 2) \\
      = & nk\left(4+ 2 r +  (r+1)\right)2^{r-1} + n2^{r} + nkr  - n(r + 2) \\
      < & 3nk\left(r+2\right)2^{r-1} + n2^{r} + nkr.\\
    \end{array}
  \end{equation*}

  For a simpler depth analysis, we assume that the input generate signals
  $y_i$ arrive delayed at a depth of $r+2$.
  The generate input signals  traverse a binary tree of depth $r-1$ and 
  the  propagate input signals traverse a binary tree of depth $r+1$
  before reaching the first multi-input generate gate, i.\ e.\  generate signals $y_i$
  become available at depth $2r+1$ and propagate signals at depth $r+1$.
  Thus, the first row of multi-input generate gates has depth
  $$3r+2 = \max \{2r+1 + 1+ r,r+1+r+1+r\},$$
  where the first term in the maximum is caused by the delayed
  generate signals $y_i$ and the second term by the propagate signals
  $x_i$ ($1\le i \le 1$).

  For the next level, the propagate signals are available at time $2r+2$, and the generate signals at time $3r+2$, and the propagate signals again arrive $r$ time units before the corresponding generate signals, so at the next level, both signals arrive $r+1$ time units later than they did before. Inductively, we know that for each level $2 \le l \le k$, the generate and propagate
  signals arrive at a depth of $(l-1)(r+1)$ more than than they did for at the first
  level. Consequently, the total depth of the adder is $(k-1)(r+1) + 3r + 2 = kr
  + 2r + k +1$.
\end{proof}

If ${\sqrt{\log n}}\in \Nset$, we can choose $r=k=\sqrt{\log n}$ and
receive the following result.
\begin{cor}
If $\sqrt{\log n}\in \Nset$, there is a multi-input generate adder
for $n$ bits with fan-out two,  size at most
$$3 n({\log n}+2\sqrt{\log n})2^{\sqrt{\log n}-1}  + n2^{\sqrt{\log n}}  + n\log n,$$  and depth
    $$ \log n + 3{\sqrt{\log n}} + 1.$$
\label{cor:depth+size-integral-sqrt-log-n}
\end{cor}

In general, ${\sqrt{\log n}}\not \in \Nset$, and we get the following result.
\begin{theorem}
 Let $n\in\Nset$. For input pairs $(x_i,y_i)$ ($i\in\{1,\dots,n\}$),
there is a circuit, computing all carry bits
 with maximum fan-out $2$,  depth at most  $$\ld n + 5\ceil{\sqrt{\ld n}} + 2.$$ The  size is at most

\begin{equation}
\label{eqn:mulit-input-generate-adder-large-n}
4 n  \ceil{\sqrt{\ld{n}}}^22^{\ceil{\sqrt{\ld {n}}}}
\end{equation}
if $n\ge 16$,  and  at most
\begin{equation}
\label{eqn:mulit-input-generate-adder-small-n}
8 n \ceil{\sqrt{\ld{n}}}^2 2^{\ceil{\sqrt{\ld {n}}}}
\end{equation}
if $n\le 15$.
\label{thm:generalized-adder}
\end{theorem}
\begin{proof}
We choose $r = k = \ceil{\sqrt{\ld n}}$ and apply Lemma~\ref{lem:depth+size-lemma}, obtaining
\begin{equation}
\begin{array}{rl}
3nk(r+2)2^{r-1} + n2^{r} +    nrk
& = n\left(3(r^2+2r)2^{r-1} + 2^{r} + r^2\right).\\
\label{eqn:multi-input-adder-size-common-bound}
\end{array}
\end{equation}
Now, if  $n\ge 16$, we have $r=k\ge 2$. Thus, we can 
 use $2r \le r^2$ and  $2^{r}+r^2 \le  r^22^{r}$ 
 to bound the right hand side by
\begin{equation*}
\begin{array}{rl}
 n\left(3\left(r^2  + r^2\right) 2^{r-1} + r^22^{r}\right)
& \displaystyle = 4nr^22^{r},\\
\end{array}
\end{equation*}
implying (\ref{eqn:mulit-input-generate-adder-large-n}).

Otherwise,   $n\le 16$,  $r=k\le 2$, $r^2\le 2r$, $r^2\le 2^r$, and the right hand side of (\ref{eqn:multi-input-adder-size-common-bound})
is bounded by
\begin{equation*}
\begin{array}{rl}
\displaystyle n\left(3\left(2r  + 2r\right)\right) 2^{r-1} + 2^{r} + 2^{r}  & \displaystyle = 8nr2^{r} \le 8nr^22^r,\\
\end{array}
\end{equation*}

The resulting depth is
\begin{equation*}
\begin{array}{rl}
kr + 2r + k +1 & = \ceil{\sqrt{\ld n}}^2  + 3 \ceil{\sqrt{\ld n}} +1 \\
               & \leq (\floor{\sqrt{\ld n}}+1)^2  + 3 \ceil{\sqrt{\ld n}} + 1\\
               & \le \ld n + 5 \ceil{\sqrt{\ld n}} +2.
\end{array}
\end{equation*}
\end{proof}

If $\sqrt{\ld n} \not \in \Nset$, the adder in Theorem~\ref{thm:generalized-adder}
is larger than necessary since it has  $n' = 2^{ \ceil{\sqrt{\ld n}}^2}> n$ inputs.
If for example $n=32$, we choose $r=k=3$ and $n'= 512$. %
Thus, if $\ceil{\sqrt{\ld n}}^2 \geq n + \ceil{\sqrt{\ld n}}$, choosing $r = \ceil{\sqrt{\ld n}} - 1$ instead still yields an adder with at least $n$ inputs and outputs and reduces the size and depth significantly. For $n = 32$, we would still obtain a $64$-input adder using this method.

The analysis can be refined further by noticing that the columns $n'$ down to $n+1$ in the augmented Kogge-Stone {\sc And}-prefix graph and the multi-input gate graph can be omitted, since they are not used for the computations of the first $n$ output bits. This reduces the size of the construction. If $n' > n$, we can omit the left half of the construction and notice that the right half of lowest row of multi-input generate gates only has $2^{r-1}$ inputs, so we can actually use $2^{r-1}$-input generate gates and reduce the depth by 1. This process can be iterated until $n' = n$, which decreases the rounding error incurred in Theorem~\ref{thm:generalized-adder}; the depth is decreased by $\ceil{\sqrt{\ld n}}^2 - \log_2 n$.

In this section, we have achieved a depth bound of $\log_2 n + \mathcal{O}(\sqrt{\log n}) = \ld n + o(\ld n)$, which is asymptotically optimal, since the lower bound is $\ld n$.

\section{Linearizing the Size of the Adder}
\label{sec:linearizing-size}
To achieve a linear size while keeping the adder asymptotically
fastest possible, we adopt a technique similar to the construction by Brent and Kung
[1982], which was first used as a size-reduction tool by Krapchenko [1967]
(see \cite[pp. 42-46]{wegener}).

\subsection{Brent-Kung Step}
\label{sec:brent-kung-step}
Brent and Kung [1982] construct a prefix graph recursively as shown in  Figure~\ref{fig:krap-reduc}.
If $n$ is a at least two, it
computes the $n/2$ intermediate results $z_{n} \circ z_{n-1}; \dots;
z_{2} \circ z_{1}$ (see Section~\ref{sec:prefix-gate-adders} for the
definition of $z_i$).
A prefix graph
for these $n/2$ inputs is used to compute the prefixes $Z_{1,2i}$
for all even indices $i\in\{1,\dots, n/2\}$.  For odd indices, the
prefix needs to be corrected by one more prefix gate as $Z_{1,2i+1} = z_{2i+1} \circ Z_{1,2i}$
($i\in\{1,\dots, n/2-1\}$).
We call this method of input halving and output correction a {\it Brent-Kung step}.
Note that the propagate signals are not needed after the correction step.
Thus, we can use reduced prefix gates (Figure~\ref{fig:pfx-to-bk-output-gate}) in the output correction step.
In these prefix gates, the left propagate signal $x_i$ is used only once. Thus, the underlying logic circuit
inherits the  fan-out of  two from the prefix graph.

\begin{figure}
  \begin{subfigure}[b]{0.49\linewidth}
  
  \centering{
    \resizebox{1\linewidth}{!}{
      \begin{tikzpicture}
        \input{krap-reduc}
      \end{tikzpicture}
    }
    \caption{Brent-Kung (reduction) step}
    \label{fig:krap-reduc}%
  }

 \end{subfigure}
 \hfill
  \begin{subfigure}[b]{0.48\linewidth}
 
  \centering{
    \resizebox{1\linewidth}{!}{
      \begin{tikzpicture}
        \input{brent-pfx}
      \end{tikzpicture}
    }
    \caption{Brent-Kung prefix graph}
    \label{fig:brent-pfx}%
  }

 \end{subfigure}
\caption{Brent-Kung Step and Prefix Graph}
\end{figure}

The Brent-Kung step  reduces the instance size by a factor of two, but it increases the depth of the construction by four and the size by $5/2n$ in terms of logic gates.

Applying these Brent-Kung steps recursively, Brent and Kung obtain a
prefix graph that has prefix gate depth $2  \ld n -1$
and logic gate depth $4 \ld n -2$.
The prefix gate depth is not optimal anymore,
but the adder has a comparatively small size of $\frac{1}{2}(5n - \ld
n - 8)$ gates, and its fan-out is bounded by two at
all inputs and gates. It is shown in Figure~\ref{fig:brent-pfx}.

Brent-Kung steps were actually known before the paper by Brent and Kung [1982]
e.g.\  they were already used in  \cite{krapchenko}.
But the Brent-Kung adder is  based solely on these steps.

\subsection{Krapchenko's Adder} \label{sec:krap}
Krapchenko's adder is a non-prefix adder
computing all carry bits with asymptotically optimal depth and linear
size. Its fan-out, on the other hand, is almost linear as well, which
makes it less useful in practice. Krapchenko's techniques can be used
to derive the following reduction, based on Brent-Kung steps.

\begin{lemma}[{\cite{krapchenko}, see   \cite[pp. 42-46]{wegener}}]
\label{lem:krap} Let $\tau \leq \ld n - 1$, then given a family of
  adders computing $k$ carry bits with depth $d(k)$, maximum fan-out
  $f(k)$ and size $s(k)$, there is a family of adders computing $n$
  carry bits with depth $d({n/2^{\tau}}) + 4\tau$,
  maximum fan-out $\max\left\{2,
    f({n/2^\tau})\right\}$ and size
  $s({n/2^\tau}) + 5n$.

  With size $s({n/2^\tau}) + {5.5n}$, we can achieve the
  same depth and a maximum fan-out of at most $\max\sset{2,
    f({n/2^\tau})}$.
\end{lemma}

\begin{proof}
  We apply $\tau$ Brent-Kung steps
  and construct the remaining adder for ${n/2^{\tau}}$ from the given adder family.
  Figure~\ref{fig:krap-reduc} shows the situation for $\tau = 1$.  The
  simple application of $\tau$ Brent-Kung steps would achieve the
  claimed depth and fan-out result, except with at most $2n$
  additional $2$-input prefix gates (because we will never add more
  prefix gates than are present in the Brent-Kung prefix graph) and
  thus with $6n$ additional logic gates.

\begin{figure}
 
  \centering{
    \resizebox{0.33\linewidth}{!}{
      \begin{tikzpicture}
        \input{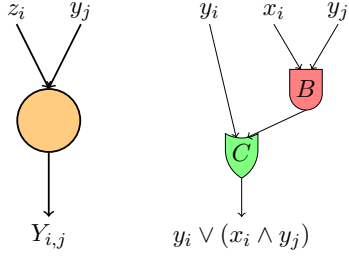}
      \end{tikzpicture}
    }
    \caption{Reduced output correction prefix gate of a refined Brent-Kung step}
    \label{fig:pfx-to-bk-output-gate}%
  }

\end{figure}

  To see that $5n$ logic gates are enough, we show that we can omit
  the propagate signal computation for the parity-correcting part of the
  Brent-Kung step.
  Such a reduced output prefix gate is shown in Figure~\ref{fig:pfx-to-bk-output-gate}.
  With this construction, note that for $i$ even, we
  have computed $(x, y) = z_i \circ \dots \circ z_1$. For $z_{i+1} =
  (y_{i+1}, x_{i+1})$, the carry bit arising from position $i+1$ is
  $c_{i+2} = x_{i+1} \vee (y_{i+1} \wedge y)$, which uses two
  gates. It follows that a Brent-Kung step uses only the propagate signals
  at the inputs. For the next Brent-Kung step, the inputs are the
  $n/2$ pairs $z_n \circ z_{n-1}; \dots; z_2 \circ z_1$, therefore we
  need three logic gates per prefix gate for the reduction step.

  Note that in Figure~\ref{fig:brent-pfx}, the propagate signal at a gate is used
  if and only if there is a vertical line from this gate to another
  prefix gate (and not to an output or repeater). These lines
  exist only in the ``upper half'' of the adder, i.\ e.\ the parts
  with depth $\leq {\ld n}$. Since parity correction occurs
  exclusively in the lower half with depth $> {\ld n}$, the
  propagate signals from parity correction steps are never used.

  As in the Brent-Kung prefix graph, $\frac{n}{2}$ repeaters can be
  used to distribute the fan-out and reduce the maximum fan-out of the
  parity-correcting gates to two (see also Figure~\ref{fig:brent-pfx}).
\end{proof}

The fact that  the refined Brent-Kung step does not require the inner adder to
provide the propagate signals, which a prefix graph adder would
provide, allows us to use the multi-input generate adder with the size
and depth bounds stated in Theorem~\ref{thm:generalized-adder}, and which
omits the last $r$ rows of {\sc And} gates (hatched gates in
Figure~\ref{fig:kogge-aug}) in the augmented Kogge-Stone {\sc
  And}-prefix graph.

Lemma~\ref{lem:krap} can be used to achieve different trade-offs. In particular,
constructions for all carry bits of size up to $n^{1+o(1)}$ can be
turned into linear-size circuits with the same asymptotic depth or
depth guarantee, since we could choose $\tau = o(1) \ld n$. This works for prefix graphs and logic circuits; for
example with $\tau = \ld \ld n$, the Kogge-Stone prefix graph will
have size $3n$, depth $\ld n + 2\ld \ld n$ and fan-out bounded by two
in terms of prefix gates  \cite{han+carlson}.

While the technique in Lemma~\ref{lem:krap} is essentially a $2$-input
prefix gate construction, the main result of \cite{krapchenko}
cannot be constructed using only prefix gates.

\subsection{Adders with Asymptotically Minimum Depth, Linear Size, and Fan-Out Two}

By combining Theorem~\ref{thm:generalized-adder} and
Lemma~\ref{lem:krap}, we get an adder of asymptotically minimum depth,
linear size and with fan-out at most two.
\begin{theorem}
  There is an adder for $n$ inputs of size bounded by ${13.5 n}$ with
  depth $$\ld n + 8 \ceil{\sqrt{\ld n}} + 6 \ceil{\ld\ceil{\sqrt{\ld n}}} + 2$$
  and maximum fan-out two.
  If $n\ge 4096$, the size can be bounded by ${9.5}n$.
  \label{all carry bits}
\end{theorem}
\begin{proof}
We apply Lemma~\ref{lem:krap} with $\tau = \ceil{\sqrt{\ld n}+ 2 \ld{\ceil{\sqrt{\ld n}}}}$
and use an adder for $n/2^{\tau}$ inputs according to Theorem~\ref{thm:generalized-adder} as an inner adder.
From the proof of Lemma~\ref{lem:krap}, we have seen that the output correction of the Brent-Kung step does not
require propagate signals from the inner adder. So the fan-out is indeed two.
Using  (\ref{eqn:mulit-input-generate-adder-small-n}), this results in an  adder of size
\begin{equation*}
\begin{array}{rl}
& \ceil{8\frac{n}{2^\tau} 2^{\ceil{\sqrt{\ld{\frac{n}{2^\tau}}}}} \ceil{\sqrt{\ld{\frac{n}{2^\tau}}}}^2} + {5.5n}\\
 \le &  \ceil{8\frac{n}{2^{\ceil{\sqrt{\ld n}}+ 2 \ld{\ceil{\sqrt{\ld n}}}}} 2^{\ceil{\sqrt{\ld{n}}}} \ceil{\sqrt{\ld{n}
}}^2} + {5.5n} \\
\le &8n+{5.5n} = {13.5 n}.
\end{array}
\end{equation*}
If $n\ge 4096$   we have $n/2^{\tau} \ge 16$ that allows us to
  apply the alternative bound (\ref{eqn:mulit-input-generate-adder-large-n}) to achieve a size bound of  $9.5n$.

The depth is
\begin{equation*}
\begin{array}{rl}
\ld{\frac{n}{2^\tau}} + 5\ceil{\sqrt{\ld \frac{n}{2^\tau}}} + 2 + 4\tau
& =\ld{n} + 5\ceil{\sqrt{\ld \frac{n}{2^\tau}}} + 2 + 3\tau \\
&\le \ld n + 8 \ceil{\sqrt{\ld n}} + 6 \ceil{\ld\ceil{\sqrt{\ld n}}} + 2,
\end{array}
\end{equation*}
where we are using
$\tau \le \ceil{\sqrt{\ld n}}+ 2 \ceil{\ld{\ceil{\sqrt{\ld n}}}}$
for the inequality.
\end{proof}

From Theorem~\ref{all carry bits}, we can easily conclude our main result
 stated in Section~\ref{sec:our_contribution}:
\maintheorem*

\section{Technology Mapping}
\label{sec:technology-mapping}
In this section we show that our construction from Theorem~\ref{all
  carry bits} can be transformed into an adder using only
\textsc{Nand}/\textsc{Nor}, and \textsc{Not} gates, which are faster
in current CMOS technologies. This increases the depth by one and the size by a small constant factor.
\begin{theorem}
  There is an adder for $n$ inputs   using only \textsc{Nand}, \textsc{Nor}, and \textsc{Not} gates.
  Its  size is bounded by ${(18+\frac{1}{3})n}$,
  the depth is at most
  $$\ld n + 8 \ceil{\sqrt{\ld n}} + 6 \ceil{\ld\ceil{\sqrt{\ld n}}} + 3,$$
  and  the maximum fan-out is two.
  \label{all carry bits inverted}
  If $n \ge 4096$, the  size is bounded by ${(15+\frac{5}{6})n}$,
\end{theorem}
In the next two sections, we show how to transform the two main components of our construction, the Brent-Kung steps and the multi-input multi-output generate gate adder, into circuits using only the desired gates. 

\subsection{Mapping  Brent Kung Steps}
\label{sec:mapping-brent-kung}
Brent-Kung steps can be implemented using \textsc{Nand}/\textsc{Not} prefix gates as shown in Figure~\ref{fig:pfx-nandnot}
in the reduction step.
Similarly, the reduced output reduction gate in Figure~\ref{fig:pfx-to-bk-output-gate}
can be realized by two \textsc{Nand} gates and one \textsc{Not} gate, i.e.\ by
 eliminating  the two rightmost gates in Figure~\ref{fig:pfx-nandnot}.
The modified prefix gates  do not increase the depth of the Brent-Kung step, and increase the  size 
by a constant factor less than $ \frac{5}{3}$.

\begin{figure}[htb]

  \centering{
    \resizebox{0.33\linewidth}{!}{
      \begin{tikzpicture}
        \input{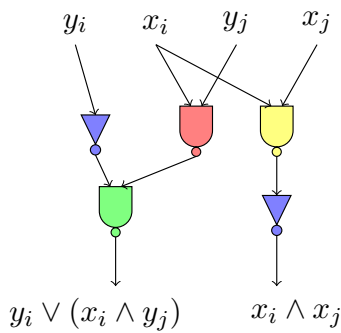}
      \end{tikzpicture}
    }
    \caption{A \textsc{Nand}/\textsc{Not} prefix gate used in the  reduction step}
    \label{fig:pfx-nandnot}%
  }

\end{figure}

\subsection{Mapping Multi-Input Multi-Output Generate Adders}
We want to transform a multi-input multi-output generate gate adder using DeMorgan's laws.
For easier understanding, we first  insert repeaters so that the gates  can be arranged in rows, such
all input signals for gates in odd rows are computed in even rows and vice versa.
This bipartite structure is already given in the  augmented Kogge-Stone {\sc And}-prefix graph (see Figure~\ref{fig:kogge-aug}).

Let us now consider a multi-input generate gate shown in Figure~\ref{fig:fo-pfx-small}.
By inserting $2^r/2$ repeaters  gates in the last row of the \textsc{And}-suffix graph, we
achieve a uniform depth of this first stage. 
The red row of \textsc{And} gates and the  final $2^{r-1}$-output \textsc{Or} already have a uniform depth.
The additional repeaters increase the size by less than  a factor of $\frac{5}{3}$.
Except for the  first row of generate gates, the depth of the generate signals equals the depth
of the propagate signals when they are merged in the red row of \textsc{And} gates.
In the first row of
generate gates, the propagate signals arrive there at depth $2r+1$, while the generate signals arrive at
depth $r-1$ (see the proof of Lemma~\ref{lem:depth+size-lemma}). Thus, if $r$ is odd, we add one additional repeater at every
generate input signal so that it arrives at an odd depth at the red level of \textsc{And} gates.
Note that we can do this without increasing the overall depth, as we already assumed that the generate signals are delayed
by $r+1$ in the proof of Lemma~\ref{lem:depth+size-lemma}.
At most $n$ repeaters are inserted this way.

Some generate gates of the multi-input generate gate adder are just buffer trees, i.e.\ blue boxes in Figure~\ref{fig:fixed-size2}.
They have depth $r-1$, which is odd if and only if the depth $r+1$ of the corresponding paths of generate signals through multi-input generate gates is odd. Thus, they preserve the bipartite structure.

Now we can use  the bipartite structure to transform the multi-input multi-output generate adder into a circuit consisting of \textsc{Nand},
\text{Nor}, and \textsc{Not} gates.
In our construction we will maintain the following characteristics.
Inputs  to an odd row, i.e.\  outputs  of an even row,  will be the original function values, while
inputs to an even row, i.e.\   outputs of an  odd row, will be the negated original function values.
We achieve this by transforming gates as follows:
Repeaters are always transformed into \textsc{Not} gates.
In odd rows, we translate
\textsc{And} gates into  \textsc{Nand} gates and \textsc{Or} gates into \textsc{Nor} gates.
In even rows, we translate
\textsc{And} gates into  \textsc{Nor} gates and \textsc{Or} gates into \textsc{Nand} gates.
If the number of rows is odd, we add one row of \textsc{Not} gates to correct the otherwise negated outputs of the adder.

Together with the $n$ repeaters that we insert behind each generate input signal if $r$ is odd,
this makes  $2n$ gates that can by accounted for by the size of the
augmented Kogge-Stone {\sc And}-prefix graph (see Figure~\ref{fig:kogge-aug}),
which is at least $3n$ if $r\ge 1$. %
Thus, the overall size of the generate adder increases by  at most a factor of $\frac{5}{3}$.

Together with the mapping of the Brent-Kung step in Section~\ref{sec:mapping-brent-kung}, this proves Theorem~\ref{all carry bits inverted}.

\section*{Conclusion}
We introduced the first full adder with an asymptotically optimum
depth, linear size and a maximum fan-out of two.  Asymptotically, this
is twice as fast and significantly smaller than the Kogge-Stone adder, which
is often considered the fastest adder circuit, as well as most other prefix graph adders.

For small $n$, Theorem~\ref{all carry bits} will not immediately improve upon existing adders.
When focusing on speed for small $n$, one would rather omit the size reduction from Section~\ref{sec:linearizing-size}.
Without the size reduction, our results in Lemma~\ref{lem:depth+size-lemma} match the depth of the Kogge-Stone adder for 512 inputs and improve on it for 2048 inputs, where $r = 3, k = 4$ yields an adder with depth 21 for our construction, but the adder of Kogge-Stone will have depth 22.

Today's  microprocessors contain adders for a few hundred bits.
However, adders for 2048 bit numbers exist already today on cryptographic chips.
Thus we expect that adders based on our ideas will find their way into hardware.

\bibliographystyle{ACM-Reference-Format-Journals}

\begin{thebibliography}{10}


\bibitem[Brent 1970]{brent} R.P. Brent. \emph{On the Addition of
    Binary Numbers.} IEEE Transactions on Computers 19.8 (1970):
  758--759.

\bibitem[Brent and Kung 1982]{brent-kung} R.P. Brent and H.-T. Kung. \emph{A
     regular layout for parallel adders.} IEEE Transactions
   on Computers 100.3 (1982): 260--264.


\bibitem[Chatterjee et al. 2006]{Chatterjee+techmap2006}
S. Chatterjee, A. Mishchenko, R. Brayton, X. Wang, and T. Kam.
\emph{Reducing structural bias in technology mapping.}
IEEE Transactions on Computer Aided Design of Integrated Circuits and Systems 25.12 (2006): 2894--2903.



\bibitem[Fich 1983]{fich} F.E. Fich. \emph{New bounds for parallel
    prefix circuits.} Proceedings of the 15th Annual ACM
  Symposium on Theory of Computing (STOC). ACM, 1983.


\bibitem[Gashkov et al. 2007]{Gashkov+ImprovingKrapchenko2007}
S.B. Gashkov, M.I. Grinchuk, and I.S.  Sergeev.
\emph{On the construction of schemes for adders of small depth}.
Diskretnyi Analiz i Issledovanie Operatsii, Ser. 1, 14.1 (2007):    27--44 (in Russian).
English translation in Journal of Applied and Industrial Mathematics 2.2, (2008): 167–-178.


\bibitem[Grinchuk 2008]{Grinchuk-ShallowCarryBit2009}
M.I. Grinchuk.
\emph{Sharpening an upper bound on the adder and comparator depths}.
Diskretnyi Analiz i Issledovanie Operatsii, Ser. 1,  15.2 (2008): 12-–22 (in Russian).
English translation in Journal of Applied and Industrial Mathematics 3.1, (2009): 61--67.

\bibitem[Han and Carlson 1987]{han+carlson}
T. Han and D.A. Carlson.
\emph{Fast Area Efficient VLSI Adders.}
8th IEEE Symposium on Computer Arithmetic (1987):  49--56.

\bibitem[Held and Spirkl 2014]{held+spirkl:2014} S. Held and S. T. Spirkl. \emph{Fast Prefix Adders for Non-Uniform Input Arrival
Times.} Algorithmica (2015).

\bibitem[Hoover et al. 1984]{Hoover+Klawe+Pippenger1984}
H.J. Hoover, M.M. Klawe,  and N.J. Pippenger.
\emph{Bounding Fan-out in Logical Networks}.
Journal of the  ACM (JACM) 31.1 (1984): 13--18.


\bibitem[Keutzer 1988]{keutzer88}
K. Keutzer.
\emph{DAGON: technology binding and local optimization by DAG matching.}
Papers on Twenty-five years of electronic design automation, ACM (1988): 617--624.


\bibitem[Knowles 1999]{knowles} S. Knowles. \emph{A family of adders.}
  Proceedings of 14th IEEE Symposium on Computer Arithmetic (1999): 277 -- 281.


\bibitem[Kogge and Stone 1973]{kogge-stone} P.M. Kogge and H.S. Stone.
  \emph{A parallel algorithm for the efficient solution of a
    general class of recurrence equations.} Computers, IEEE
  Transactions on Computers C-22.8 (1973): 786--793.


\bibitem[Krapchenko 1967]{krapchenko} V.M. Krapchenko.
 \emph{Asymptotic estimation of
 addition time of a parallel adder}.
Problemy Kibernetiki 19 (1967): 107--122 (in Russian).
English translation in System Theory Res. 19 (1970): 105--122.

\bibitem[Krapchenko 2007]{krapchenkoLB} V.M. Krapchenko.
  \emph{On Possibility of Refining Bounds for the Delay of a Parallel Adder}.
Diskretnyi Analiz i Issledovanie Operatsii, Ser. 1, 14.1 (2007): 87--93.
English translation in Journal of Applied and Industrial Mathematics 2.2 (2008):  211--214.



\bibitem[Ladner and Fischer 1980]{ladner} R.E. Ladner and M.J. Fischer.
  \emph{Parallel prefix computation.} Journal of the ACM
  (JACM) 27.4 (1980): 831--838.


\bibitem[Lupanov 1962]{LupanovBoundedBranching62}
O.B. Lupanov.
\emph{A class of schemes of functional elements}.
Problemy Kibernetiki  7  (1962):  61--114.
English translation in Problems of Cybernetics 7 (1963): 68--136.

\bibitem[Ofman 1962]{Ofman1962}
Y.P. Ofman.
\emph{The algorithmic complexity of discrete functions}.
Doklady Akademii Nauk SSSR 145.1 (1962): 48--51.
English translation in Soviet Physics Doklady 7 (1963): 589--591.

 \bibitem[Rautenbach et al. 2008]{code-trees} D. Rautenbach, C. Szegedy and
   J. Werber. \emph{On the cost of optimal alphabetic code trees
   with unequal letter costs.} European Journal of Combinatorics 29.2
   (2008): 386-394.

\bibitem[Sergeev 2013]{sergeev} I. Sergeev. \emph{On the complexity of
  parallel prefix circuits.} Electronic Colloquium on Computational
  Complexity (ECCC). Vol. 20. 2013.



 \bibitem[Sklansky 1960]{sklansky} J. Sklansky. \emph{Conditional-sum addition
     logic.} Electronic Computers, IRE Transactions on 2 (1960):
     226-231.

\bibitem[Wegener 1987]{wegener} I. Wegener. \emph{The complexity of Boolean
    functions}. Wiley-Teubner (1987).

\bibitem[Werber et al. 2007]{bonnlogic} J. Werber, D. Rautenbach and
  C. Szegedy. \emph{Timing optimization by restructuring long
    combinatorial paths.} Proceedings of the 2007 IEEE/ACM
  international conference on Computer-aided design (2007): 536--543.

\end{thebibliography}

\end{document}